\documentclass[letterpaper]{article} 
\usepackage{times}  
\usepackage{helvet} 
\usepackage{courier}  
\usepackage[hyphens]{url}  
\usepackage{graphicx} 
\usepackage{natbib}  
\usepackage{caption} 
\usepackage{color}
\usepackage{tikz}
\usepackage{amsthm,amssymb,amsmath}
\usepackage[switch]{lineno}
\usetikzlibrary{positioning,chains,fit,shapes,calc}
\usepackage{ifthen}

\newtheorem{definition}{Definition}
\newtheorem{example}{Example}
\newtheorem{proposition}{Proposition}
\usepackage{booktabs}
\usepackage{adjustbox}
\usepackage{epigraph}

\usepackage{longtable}
\usepackage{fullpage}

\newtheorem{challenge}{Research Challenge}

\newcommand{\ignore}[1]{}

\title{The Affiliate Matching Problem: On Labor Markets where Firms are Also Interested in the Placement of Previous Workers\footnote{This is working paper of early-stage research.  We realize the matching literature is expansive and multi-disciplinary, and we welcome connections between published research and our model. Corresponding author: sdooley1@cs.umd.edu}}
\author {
    Samuel Dooley \qquad \qquad John P. Dickerson\\
    \small Computer Scicence Department\\
    \small University of Maryland
}
\date{}

\begin{document}

\maketitle

\begin{abstract}
 In many labor markets, workers and firms are connected via affiliative relationships.  A management consulting firm wishes to both accept the best new workers but also place its current affiliated workers at strong firms.  Similarly, a research university wishes to hire strong job market candidates while also placing its own candidates at strong peer universities.  We model this \emph{affiliate matching} problem in a generalization of the classic stable marriage setting by permitting firms to state preferences over not just which workers to whom they are matched, but also to which firms their affiliated workers are matched. Based on results from a human survey, we find that participants (acting as firms) give preference to their own affiliate workers in surprising ways that violate some assumptions of the classical stable marriage problem.  This motivates a nuanced discussion of how stability could be defined in affiliate matching problems; we give an example of a marketplace which admits a stable match under one natural definition of stability, and does not for that same marketplace under a different, but still natural, definition. We conclude by setting a research agenda toward the creation of a centralized clearing mechanism in this general setting.

\end{abstract}

\epigraph{{\it My loyalty for my student made me choose the best place for him.}}{{Survey Respondent \#103}}
\vspace{-2em}
\epigraph{{\it Strange here because I opted to go with my third choice Ryan over getting my second choice ... an emotional reaction because he is a graduate of my school. }
}{{Survey Respondent \#6}}
\vspace{-2em}
\epigraph{{\it I would not play favorites with my student. It would not be fair.}}{Survey Respondent \#136}

\section{Introduction}\label{sec:intro}

Markets are systems that empower interested parties---humans, firms, governments, or autonomous agents---to exchange goods, services, and information.  In some markets, such as stock and commodity exchanges, prices do all of the ``work'' of matching supply and demand.  Yet, due to logistical or societal constraints~\citep{Roth07:Repugnance}, many markets cannot rely \emph{solely} on prices to match supply and demand.  Examples include school choice~\citep{Abdulkadiroglu03:School},
medical residency~\citep{Roth99:Redesign}, online dating~\citep{Bragdon10:Matching}, advertising~\citep{Edelman07:Internet}, deceased-donor organ allocation~\citep{Bertsimas13:Fairness}, online labor~\citep{Arnosti14:Managing}, housing allocation~\citep{Bloch13:Markovian}, job search~\citep{Das10:When}, refugee placement~\citep{Delacretaz19:Matching}, and barter markets such as kidney exchange~\citep{Dickerson15:FutureMatch,Manlove15:Paired}.


We focus on two-sided \emph{matching problems} where agents on one side of a market are allocated to agents on the other based solely on ranked preferences over potential matching outcomes.  Traditionally, agents express a preference ordering over other agents, e.g., as in the classical stable marriage problem due to~\citet{Gale62:College}.
Yet, agents may implicitly have preferences over their own match outcome, but \emph{also} over the match outcome of others---though they are seldom able to express that latter class of preference explicitly.  As an explicit example, in an academic hiring market, universities both place candidates into the job market (via graduation) and take candidates out of the job market (via hiring).  Universities care not only about the quality of candidate that they hire, but also about the match outcome(s) of their graduating candidates.  Thus, expression of such preferences by a firm (\emph{n{\'e}} employer) necessarily extends beyond a simple rank ordering of workers (\emph{n{\'e}} candidates) on the opposing side.

\noindent\textbf{Our contributions.}  In this paper, we study how a class of explicit complex preferences over match outcomes impacts a matching market.  Specifically, we:
\begin{itemize}
    \item Initiate the study of the \emph{affiliate matching} problem within a generalization of the classical stable marriage setting;
    \item Provide experimental evidence that humans (acting as firms) on one side of a market do ``act differently'' with respect to stable marriage when expressing preferences over match outcomes when worker affiliations are present;
    \item Draw on joint intuition from the matching literature as well as our human experiment to define appropriate notions of a standard desideratum in matching settings, \emph{stability}, and show that different natural definitions of stability result in the differing existence of stable matchings; and
    \item Discuss four concrete research challenges for future scholarship in this novel matching setting. 
\end{itemize}

\noindent\textbf{Placement in the literature.}
The stable marriage problem, as presented by~\citet{Gale62:College}, operates a centralized matching scheme where agents on one side of a market have strict \emph{ordinal preferences} over (subsets) of agents on the other side.  Here, \emph{stability} is a central desideratum to be achieved in a final match outcome (see work due to~\citet{Roth82:Economics} for additional motivation); we continue that focus in this paper, albeit (necessarily) by defining new notions of stability to fit our general model. 

Many variants of stable matching problems have been explored, including: additional constraints such as capacity~\citep{Gusfield89:Stable}, where agents on one side of the market may be matched to more than one but at most some limit (with application found in the National Residency Matching Program (NRMP)~\citep{Roth99:Redesign}); preference orderings that are incomplete~\citep{Gusfield89:Stable} or are not strict~\citep{Gusfield89:Stable,Irving94:Stable}; non-bipartite markets as in the case of the stable roommates problem for additional coverage; and many others.  We direct the reader to a recent comprehensive overview due to~\citet{Manlove13:Algorithmics}.  In our work, we extend the preference language beyond that found in these and other traditional settings (described in greater detail next in Section~\ref{sec:model}).

The AI community has played a central role in understanding the computational complexity impacts of various constraints and generalizations in matching problems~\citep{Iwama99:Stable,Manlove02:Hard,Liu14:Internally,Perrault16:Strategy-proofness,Delorme19:Mathematical}.  While basic forms of matching are often solvable in polynomial time (see, for example, the seminal work due to~\citet{Gale62:College} as well as~\citet{Irving85:Efficient} for two well-known examples), generalized versions often required the use of sophisticated optimization approaches such as constraint programming~\citep{Prosser14:Stable,Manlove17:Almost-stable}, SAT~\citep{Drummond15:SAT,Endriss20:Analysis}, answer set programming~\citep{Erdem20:General}, integer programming~\citep{Delorme19:Mathematical}; we direct the reader to a recent survey due to~\citet{Geist17:Computer-aided} for an in-depth review of computational approaches to clearing stable matching problems.  In Section~\ref{sec:stability}, we follow this trend and provide a general integer-programming-based approach to finding stable matchings (or proving one does not exist) in our context.

\section{Model of Affiliate Matching}\label{sec:model}

We now set up the notation, intuition, and initial observations of modeling the affiliate matching problem.

We are motivated by scenarios where the two sides of the market are interconnected via affiliations, like faculty hiring markets, certain consulting labor markets, etc. We want to study how these affiliations influence the dynamics of this matching market; so we first define the problem.

\subsection{Affiliate Matching}

We introduce a new matching model, which we call {\bf affiliate matching}, which assumes there are two, disjoint sides of a market where each side has \emph{affiliations} with the other. In the case of the faculty matching market, an applicant and a university have an affiliation if the applicant attended that university for their graduate studies. In some labor markets like management consulting, an applicant and employer have an affiliation if the applicant was previously employed by that employer.
Formally, let $A=\{a_i\}_{i=1}^n$ be a set of applicants and $E=\{e_j\}_{j=1}^m$ be a set of employers. For each employer, $e$, let us denote the set of applicants affiliated with $e$ with the \emph{ordered} set $R_{e}$. 
Denote the cardinality of $R_{e}$ as $r_e$. 

We define matches as in the standard bipartite matching problem. A match $\mu$ is a function $\mu:A \bigsqcup E \to A \bigsqcup E$ such that for each element $a\in A$, we have $\mu(a)\in E$, and for each element $e\in E$, we have $\mu(e)\in A$. Further, we note that $\mu(a) =e$ if and only if $\mu(e)=a$.

Based on the motivating scenarios, it makes sense to have the preference languages for each side of the market vary slightly. The applicants will have their standard preference language, as in stable marriage. We denote $\prec_{a}$ to be the preference of applicant $a$ where $\alpha\prec_{a}\beta$ means that outcome $\beta$ is preferred to outcome $\alpha$ by candidate $a$. 
The employers however will have a modified preference language that indicates where they would like to see their affiliates be matched. 
Each employer has preference orderings $\prec_{e}$ over the tuples of $A\times E^{r_e}$. 
For employer $e$, we interpret an outcome $(a,\gamma_1,\dots,\gamma_{r_e})\in A\times E^{r_e}$ to mean that the employer $e$ is matched with applicant $a$ and would like the $k$-th member of $R_e$ to be matched with employer $\gamma_k$.
We note that $\prec_e$ will not be a complete and total ordering over $A\times E^{r_e}$ as there may be inconsistencies present in a tuple.
For example, one cannot prefer to be matched with their own affiliate and have their affiliate be matched elsewhere.
For this paper, let us assume that preferences are expressed as strict, for ease of introduction. 

\definecolor{myblue}{RGB}{80,80,160}
\definecolor{myred}{RGB}{160,80,80}

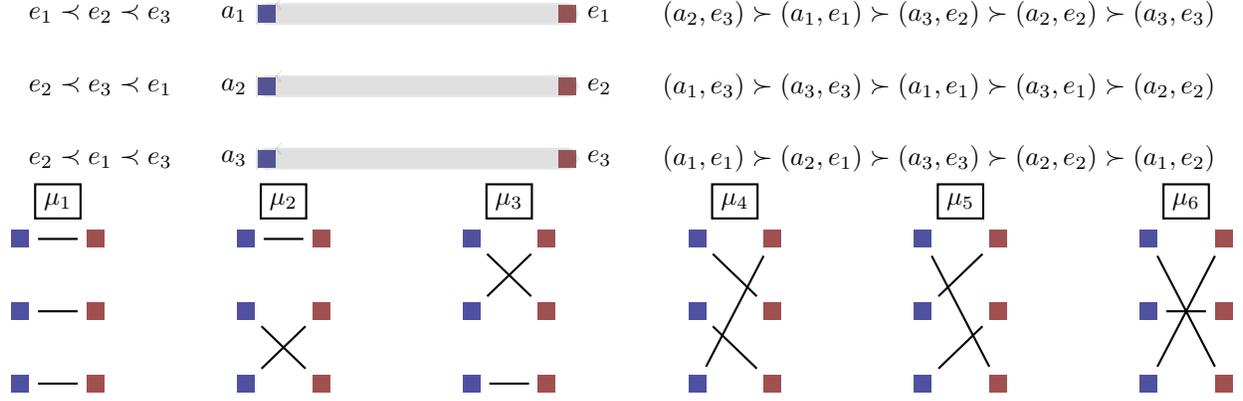
\begin{figure*}
    \centering
    \begin{tikzpicture}[thick,
      asnode/.style={fill=myblue},
      esnode/.style={fill=myred},
      every fit/.style={ellipse,draw,inner sep=-2pt,text width=2cm},
      ->,shorten >= 3pt,shorten <= 3pt
    ]
    
    \begin{scope}[start chain=going below,node distance=7mm]
    \foreach \i in {1,2,3}
      \node[asnode,on chain] (a\i) [label=left: $a_\i$] {};
    \end{scope}
    \node[left= of a1] {$e_1 \prec e_2 \prec e_3$};
    \node[left= of a2] {$e_2 \prec e_3 \prec e_1$};
    \node[left= of a3] {$e_2 \prec e_1 \prec e_3$};
    
    \begin{scope}[xshift=4cm,start chain=going below,node distance=7mm]
    \foreach \i in {1,2,3}
      \node[esnode,on chain] (e\i) [label=right: $e_\i$] {};
    \end{scope}
    \node[right= of e1] {$(a_2,e_3) \succ (a_1,e_1) \succ (a_3,e_2) \succ (a_2,e_2) \succ (a_3,e_3)$};
    \node[right= of e2] {$(a_1,e_3) \succ (a_3,e_3) \succ (a_1,e_1) \succ (a_3,e_1) \succ (a_2,e_2)$};
    \node[right= of e3] {$(a_1,e_1) \succ (a_2,e_1) \succ (a_3,e_3) \succ (a_2,e_2) \succ (a_1,e_2)$};
    
    \draw[draw=gray, rounded corners, thick, fill=gray!80!black, opacity=0.2]
        (a1.north) -| (e1.south east) -- (a1.south east) -| 
        (a1.west) |- (a1.north);
    \draw[draw=gray, rounded corners, thick, fill=gray!80!black, opacity=0.2]
        (a2.north) -| (e2.south east) -- (a2.south east) -| 
        (a2.west) |- (a2.north);
    \draw[draw=gray, rounded corners, thick, fill=gray!80!black, opacity=0.2]
        (a3.north) -| (e3.south east) -- (a3.south east) -| 
        (a3.west) |- (a3.north);
    \end{tikzpicture}
    \begin{tikzpicture}[thick,
      asnode/.style={fill=myblue},
      esnode/.style={fill=myred},
      every fit/.style={ellipse,draw,inner sep=-2pt,text width=2cm},
      ->,shorten >= 3pt,shorten <= 3pt
    ]
    \foreach \j in {0,1,2,3,4,5}
    {%
        \pgfmathtruncatemacro{\jpo}{\j+1}
        \pgfmathsetmacro{\aj}{3*\j}
        \pgfmathsetmacro{\ej}{3*\j+1}
        \pgfmathsetmacro{\mup}{3*\j+0.5}
        \node[draw] at (\mup,0.5) (mu\j)  {$\mu_\jpo$};
        \begin{scope}[xshift=\aj cm, start chain=going below,node distance=7mm]
        \foreach \i in {0,1,2} 
          {
          \node[asnode,on chain] (a\i\j) {};
          }
        \end{scope}
        \begin{scope}[xshift=\ej cm,start chain=going below,node distance=7mm]
        \foreach \i in {0,1,2}
          {
          \node[esnode,on chain] (e\i\j) {};
          }
        \end{scope}
        \pgfmathsetmacro{\jj}{1\j}
        \pgfmathsetmacro{\jj}{1\j}
        \ifthenelse{\j < 2}{\pgfmathsetmacro{\eoj}{0}}{
        \ifthenelse{\j < 4}{\pgfmathsetmacro{\eoj}{1}}{\pgfmathsetmacro{\eoj}{2}}
        };
        \ifthenelse{\j=0 \OR \j=5}{\pgfmathsetmacro{\etj}{1}}{
        \ifthenelse{\j=1 \OR \j=3}{\pgfmathsetmacro{\etj}{2}}{\pgfmathsetmacro{\etj}{0}}
        };
        \ifthenelse{\j=1 \OR \j=4}{\pgfmathsetmacro{\ehj}{1}}{
        \ifthenelse{\j=0 \OR \j=2}{\pgfmathsetmacro{\ehj}{2}}{\pgfmathsetmacro{\ehj}{0}}
        };
        \path[-] (a0\j) edge (e\eoj\j);
        \path[-] (a1\j) edge (e\etj\j);
        \path[-] (a2\j) edge (e\ehj\j);
    }
    \end{tikzpicture}       
    \caption{Above, a completely expressed example marketplace between employers, each with one affiliate. Affiliations are denoted with a gray connection between $a_i$ and $e_i$ for $i=1,2,3$. Below, the 6 possible matchings are enumerated.}
    \label{fig:example marketplace}
\end{figure*}
\begin{example}
Let us look at an example in Figure \ref{fig:example marketplace}.
In this example, we have three agents on each side of the market: three applicants $a_1$, $a_2$, $a_3$, and three employers $e_1,e_2,e_3$. Each employer only has one affiliate, $r_e=1$ for all $e$, and $a_i$ is the affiliate of $e_i$ for all $i$ (denoted with the gray connections between the two sides). Each agent's preferences are expressed in the preference language defined above.
\end{example} 

We observe something unique about this model that has particular salience for affiliate matching: we can relate the preference profiles of the agents to the structure of the possible matchings. 
This is true in the stable marriage problem, but has little utility---an agent's preference language only permits expression of a preference over their own match. In this case, like for $a_1,a_2,a_3$, we can relate their profile to the six possible matchings $\mu_i$ for $i=1,\dots 6$. For instance, there is a one-to-one relationship between  $a_1$'s preference profile $e_1\prec e_2 \prec e_3$ and $\mu_1,\mu_2\prec \mu_3,\mu_4 \prec \mu_5,\mu_6$. We can do a similar thing for $e_1$ and rewrite their preference profile as $\mu_5 \succ \mu_1,\mu_2 \succ \mu_4 \succ \mu_3 \succ \mu_6$. Rewriting the preferences like this demonstrates the connection between a preference profile for an employer and the overall structure of its preferred matches. We observe that the employer preference language enables expression over {\it portions} of the final matching structure. s

\subsection{Connection to Other Preference Profiles}

When considering the affiliate matching problem, it appears to be a superficial change in the preference language for the employers---matches are still bipartite and employees still have their same preference language. However, this slight change brings about the question: \emph{How do employers combine their thoughts about applicants with their thoughts about their fellow employers?} It is natural to assume that each employer has some internal preference over which applicant they would like to hire. Also, they may have an ordering of the fit of the other employers for their affiliates. They may use this to articulate their preference that their affiliates be matched at higher ranking employers. It is not immediately determined how to convert two preference profiles over applicants and other employers into the preference language of tuples over $A\times E^{r_e}$. 

To formalize, consider an employer $e$. The employer may have preferences over the applicants to which they are matched, $\overline{\prec_e}$, and also preferences over their fellow employers, $\hat{\prec_e}$. They ultimately must derive a strategy $\alpha$ which takes these two profiles and outputs $\prec_e$, i.e., $\alpha(\hat{\prec_e},\overline{\prec_e})=\prec_e$. In the case of $r_e=1$, there are $((m-1)(n-1)+1)!$ possible choices for the strategy $\alpha$. Of course, a good number of them wouldn't be rational, as we can see in this example. 

\begin{example}\label{expl: consistency}
Consider a marketplace with three applicants (Ryan, Alex, and Taylor), and three employers (the Bear Mountain University (BMU), Littlewood University (LU), and West Shores University (WSU)). Assume that BMU has this preference over applicants: 
$$\text{Alex} \quad \overline{\succ} \quad \text{Ryan} \quad \overline{\succ} \quad \text{Taylor};$$ 
and BMU had this preference over employers: 
$$\text{LU} \quad \hat{\succ} \quad \text{BMU} \quad \hat{\succ} \quad \text{WSU}.$$ 
Further assume that $R_{\text{BMU}} = \{\text{Ryan}\}$. Here are some possible combination strategies $\alpha$:
\begin{enumerate}
    \item (A, LU) $\succ$ (A, WSU) $\succ$ (R, BMU) $\succ$ (T, LU) $\succ$ (T, WSU)
    \item (A, LU) $\succ$ (T, LU) $\succ$ (R, BMU) $\succ$ (A, WSU) $\succ$ (T, WSU)
    \item (T, WSU) $\succ$ (T, LU) $\succ$ (A, WSU) $\succ$ (A, LU) $\succ$ (R, BMU)
    \item (A, LU) $\succ$ (R, BMU) $\succ$ (A, WSU) $\succ$ (T, LU) $\succ$ (T, WSU)
    \item \dots
\end{enumerate}
The first two examples could be considered rational on the part of the employer $e$. We can interpret (1) as an employer who prefers being matched to their preferred applicant more than ensuring a good placement for their affiliate. The second strategy (2) does just the opposite---prioritizing the placement of their affiliate. It is difficult to label the third strategy (3) as rational, and finally the fourth strategy (4) is somewhere between (1) and (2).
\end{example}

\subsection{Consistent Employer Profiles }

The first combination strategy in Example \ref{expl: consistency} is worth further discussion. 
We observe that this strategy corresponds with a rational agent who has no regard for the matching of their affiliate. In the most extreme sense, this affiliate-agnostic agent in the Example \ref{expl: consistency} setting might randomly choose\footnote{With a relaxation of strict orderings, this could be expressed via ties/indifference.} from the set of four orderings 
$$\text{(A, ?)} \succ \text{(A, ?)} \succ \text{(R, BMU)} \succ \text{(T, ?)} \succ \text{(T, ?)}$$
where `?' is randomly chosen in $\{$LU,WSU$\}$.

One property of this affiliate-agnostic agent's preference profile is that the first element in every tuple of $\succ$ preserves the order of their $\overline{\succ}$ profile. Therefore, we have the following definition.

\begin{definition}
An employer's preference profile $\succ$ is said to be {\bf consistent} (with respect to $\overline{\succ}$) if the ordering of the first element of each tuple of $\succ$ preserves the ordering of $\overline{\succ}$.
An employer is said to be {\bf affiliate-agnostic} if their preference profile are consistent.
\end{definition}

If all agents were affiliate-agnostic with consistent preferences, then it might be fair to conclude that the affiliate matching model is not interesting and just reduces back to the stable marriage problem. However, it is reasonable to believe that not all employers, and in fact most employers with affiliates, would not have consistent preference profiles most of the time. If employers do not admit affiliate-agnostic preferences, then we should conclude that the affiliate matching problem is novel and worth modeling in a new way.
In the next section, we devise a crowdsourced survey to study if employers are actually affiliate-agnostic.

\section{Experimental Evidence in Support of the Affiliate Matching Model}\label{sec:survey}

We hypothesize that: (1) there is such variation in how employers combine their preferences over both sides of the market to maximize their utility, and (2) the affiliation induces the employer to prefer for their affiliates to be matched to more preferred employers. We study these two claims through the use of a crowdsourced survey. 

\subsection{Survey Design}
We created a survey about a faculty hiring program, and we administered using the crowdsource platform Cint. We presented the participants with a series of questions like in Example \ref{expl: consistency}: {\it Assume you are BMU in a market with three agents on either side, and your affiliate is Ryan. If you assume your preferences over schools are $\hat{\prec}$ and your preferences over applicants are $\overline{\prec}$, what is your preference profile $\prec$?}

The structure of the survey was designed as follows. There was a short introduction section which presented the faculty hiring problem and primed the participant to consider themselves to be one of the universities. Then, there was a topic-related test which asked the participant to match the text based description of a match to the visual representation thereof. This was used to screen out participants who were just clicking-through the survey. Next, the concept of affiliates was introduced and the five possible matchings were enumerated, as in Figure \ref{fig:example marketplace} where $\mu_1$ and $\mu_2$ have the same outcome if you're BMU. We then randomly assigned participants to be primed to believe that prioritizing their affiliate's matching would ultimately lead to a better outcome for their university in the future. Part of this priming language was: {\it So, keep in mind that if Ryan is placed at a higher ranked university, then your university will be viewed better and could hire better candidates in the future.}

Finally, the participants were asked to rank the five possible matchings based on randomly assigned $\hat{\succ}$ and $\overline{\succ}$, as in Example \ref{expl: consistency}. 
We only aimed to test how the employer, BMU, would behave based on where it was placed in $\hat{\succ}$, and also based on where its student, Ryan, was placed in $\overline{\succ}$. Therefore, there are a total of 9 different scenarios. The first scenario places Ryan as the top choice in $\overline{\succ}$ and BMU as the top choice in $\hat{\succ}$. (The order of the remaining agents was randomized.) The second scenario keeps BMU as the top choice, but makes Ryan the second choice. The third scenario again keeps BMU as the top choice, and makes Ryan the third choice. The fourth through sixth scenarios maintain the same pattern with BMU as the second choice. Similarly, the seventh through ninth scenarios have BMU as the third choice.
In the survey, the participants were shown seven of these nine scenarios in a random order.
The participants were also given the opportunity to explain why they chose that particular ordering.

In designing the survey in this way, we had two main {\bf research questions}:
\begin{enumerate}
    \item How common or infrequent are the different possible prioritization strategies $\alpha(\hat{\prec_e},\overline{\prec_e})=\prec_e$?
    \item Do participants care about the placement of their affiliate?
\end{enumerate}

The survey design was IRB reviewed and subsequently determined to be exempt from IRB oversight because the data were crowdsourced and we did not collect any personally identifiable information or demographics of the participants.
The entirety of the survey protocol is listed in the Appendix. 

\subsection{Survey Results}

\newcommand{\medianTime}{19.8 minutes}
\newcommand{\totalParticipants}{200}
\newcommand{\totalFinal}{154}

We used the Cint crowdsourcing platform to recruit participants. The only requirement for participation was to be English-speaking and be located in the United States. We did also filter out participants who did not pass our click through, comprehension test. The final data included \totalParticipants{} participants which we further filtered down to \totalFinal{} participants after analyzing the participants responses for any further click-throughs or participants who self-identified as not-understanding the survey. The median time to take the survey of these participants was \medianTime{}. We paid Cint \$3.05 for each of the \totalParticipants{} responses.

  \begin{table}[]
  \centering
  \begin{adjustbox}{max width=\linewidth}
\begin{tabular}{@{}lccccccccc@{}}
\toprule
                                                                & \multicolumn{9}{c}{Scenario}                                            \\ 
                                                                & 1      & 2     & 3     & 4      & 5     & 6     & 7     & 8     & 9     \\\midrule
Top 1 & 80\% & 49\% & 40\% & 47\% & 50\% & 42\% & 38\% & 42\% & 38\%\\
Top 2 & 48\% & 26\% & 24\% & 27\% & 25\% & 18\% & 20\% & 22\% & 18\%\\
Top 3 & 19\% & 20\% & 13\% & 23\% & 18\% & 12\% & 11\% & 10\% & 11\%\\
Top 4 & 18\% & 15\% & 12\% & 14\% & 13\% & 10\% & 9\% & 9\% & 9\%\\
Top 5 & 18\% & 15\% & 12\% & 14\% & 13\% & 10\% & 9\% & 9\% & 9\%\\
 \bottomrule
\end{tabular}
\end{adjustbox}
  \caption{Percentage agreement for {\it most common} Top $k$ elements of the of the preference profile across the nine scenarios. The `Top 1' condition reports the percentage of respondents who agreed with the most common outcome which they {\it most prefer}. The `Top 2' condition reports the percentage of respondents who agreed with the most common outcome for the {\it most and second most preferred}. `Top 5' reports the percentage of respondents who agreed on the most common full preference profile $\succ$. }
  \label{tab:agreement}
\end{table}

For our first research question, we ask whether the respondents were responding randomly to the different scenarios. We observe qualitatively that there is relatively strong agreement on the entire preference profiles, though that diminishes as BMU becomes less desirable. For instance, if BMU and Ryan are both in the top tier, 18\% of respondents agreed on the \emph{full} preference profile. This diminished to just 9\% when BMU and Ryan were the lowest ranked. However, agreement on the entire preference profile is a strict condition. 

We found anecdotally that participants' strategies for the survey response were to identify the top two or three outcomes they preferred and had less careful consideration for the 4th and 5th outcomes. Therefore, we also report agreement on the preference orderings for the first one, two, and three outcomes. Results are reported in Table \ref{tab:agreement}.

When we ask how common are the most preferred top three elements of the preference profile, we observe that that when BMU and Ryan are both top tier, agreement increases slightly from 18\% to 19\%. It increases to 48\% when considering agreement on the ordering of the top two, and finally achieves 80\% agreement for the most common top choice in this scenario. 

We observe similar qualitative behavior for the other scenarios; when BMU and Ryan are in the bottom tier, the most common top choice is only chosen by 38\% of respondents. 

We conclude from these results that there is structure in the survey responses, with some particular strategies, $\alpha$, being more common than others, and there not being one governing strategy for any of the scenarios which most respondents can agree is proper.

  \begin{table}[]
  \centering
  \begin{adjustbox}{max width=\linewidth}
\begin{tabular}{@{}lccccccccc@{}}
\toprule
                                                                & \multicolumn{9}{c}{Scenario}                                            \\ 
                                                                & 1      & 2     & 3     & 4      & 5     & 6     & 7     & 8     & 9     \\\midrule
\begin{tabular}[c]{@{}l@{}}Consistency\\ Frequency\end{tabular} & $26\%$ & $7\%$ & $4\%$ & $14\%$ & $4\%$ & $4\%$ & $9\%$ & $2\%$ & $1\%$ \\ \bottomrule
\end{tabular}
\end{adjustbox}
  \caption{Percentage of full preference profiles which were consistent for each scenario.}
  \label{tab:consistency}
\end{table}

For our second research question, we ask the question: how common is it for the respondents to choose a preference profile which is consistent? Recall that for any given $\overline{\succ}$ and $\hat{\succ}$, there are only four preference profiles that are considered consistent. We report the percentage of responses which are consistent for each of the nine scenarios in Table \ref{tab:consistency}. Here we see that the first and fourth scenarios have higher percentages of consistent responses than the others. One might also say that the seventh scenario also has a higher percentage, compared to its related scenarios eight and nine. The first, fourth, and seventh scenarios all correspond to the case where Ryan is the top tier student. This indicates that consistency is generally more common in these scenarios when controlling for the location of BMU. Thus, when an employer's student is ranked lower, there is a higher likelihood of having an inconsistent preference profile.

We also analyze how frequently individual respondents reported a consistent profile. We see that a total of 58 out of \totalFinal{} respondents reported a consistent preference at least once, 17 did so twice, seven did so three times, and no one did more than three. Interestingly, four of the seven who reported three consistent profiles did so on the 1st, fourth, and seventh scenarios. This indicates there actually is some structure about Ryan being the top student and the consistency of the profile. 

Recall that part of the survey design included a priming for respondents to consider the quality of match of their student as equally as important as their own match. We also found that the priming had no impact on the frequency of consistent responses. 

These results indicate two things: (1) consistency is dependent on the given scenario as dictated by $\overline{\succ}$ and $\hat{\succ}$, and (2) consistency is relatively infrequent in the majority of the cases.

\subsection{Limitations}

We briefly discuss the limitations of this experiment. Primarily, this was a survey experiment with non-experts in graduate hiring. Therefore, there might be some desirability effects in the response data, and the noise is likely higher than it would have been if we had used domain experts. It is challenging to find experts in this field who may better understand how competing priorities are dealt with when consider the match of a university and the match of that university's affiliates. It would be very interesting to conduct a similar survey with that population.

Further, the data quantity is limited. The survey would be enhanced with more responses to test perform a Chi-squared test about the randomness of the responses. However, we have strong evidence that there is structure in these responses and respondents were (generally) not just clicking through.

Finally, we tried to test about how priming the respondent might impact their response. We did not find any significant difference in the two group's consistency, which could either be a poor priming model or that there is no effect in this crowdsourced survey design.

\section{%
What is Stability in Affiliate Matching?%
}\label{sec:stability}

With any new model for a market, it is important to get the definition of stability correct. In normal stable marriage, a matching is stable if there does not exist a blocking pair $(a,e)$ which would each prefer each other over their current match. In affiliate matching, it is clear when an applicant might want to deviate from their match, as their preference language in this model is the same as in the stable marriage problem. An applicant $a$ would only want to leave a marketplace with match $\mu$ if they could find a consenting individual $e$ such that $\mu(a)\prec_a e$.

In the affiliate matching problem, the employer's preference language includes preferences over their affiliate's match. This detail provides a level of nuance which we will discuss in this section. 

Let us take the example in Figure \ref{fig:example marketplace}. Consider match $\mu_1$ and employer $e_1$. For this employer, the current match is its second most preferred outcome. However, employer $e_1$ would prefer to be matched with applicant $a_2$ as long as applicant $a_1$ is matched with employer $e_3$. Further, applicant $a_2$'s most preferred match is with employer $e_1$. However, if $a_2$ and $e_1$ left the marketplace to be matched together, $e_1$ cannot guarantee that $a_1$ will be matched with $e_3$, which would result in its most preferred outcome. If, on the other hand, as a consequence of $a_2$ and $e_1$ leaving the marketplace, it happened that $a_1$ was matched with $e_2$, that would leave $e_1$ worse off than if it stayed under $\mu_1$. So in the face of this uncertainty, the essential question is: would $e_1$ leave this marketplace? Without knowledge of the second order effects of its decision to leave, $e_1$ must make a complex decision. As mechanism designers, this presents an interesting challenge and one which we will explore with several notions of stability in this section. 

\subsection{Greedy Stability}

Under the scenario where $e_1$ is currently matched under $\mu_1$, they may choose to leave the marketplace to be matched with $a_2$ solely because a match with $a_2$ may produce a better outcome from them. Formally, this can be defined by saying that $e_1$ would deviate if there exists some match with a consenting applicant which could improve their outcome, regardless of whether that match is the final result. We put this formally with a definition of a greedy blocking pair. 

\begin{definition}
A match $\mu$ has a {\bf greedy blocking pair} $(a,e)$ if $e \succ_a \mu(a)$ and there exists $\gamma\in E^{r_e}$ such that $(a,\gamma) \succ_e (\mu(e),\mu(R_e))$. \\
Equivalently, a match $\mu$ has a {\bf greedy blocking pair} $(a,e)$ if there exists a match $\mu'$ such that the following hold: (1) $\mu'(a)=e$, (2) $e \succ_a \mu(a)$, and (3) $(a,\mu'(R_e)) \succ_e (\mu(e),\mu(R_e))$.
\end{definition}

We can then define stability as the absence of a greedy blocking pair.

\begin{definition}
A matching $\mu$ is {\bf greedily stable} if there does not exists a greedy blocking pair $(a,e)$.
\end{definition}

There are two interesting properties of this definition of stability: it reduces to the stable marriage problem under consistent preferences, but it may have no stable match if the preferences are allowed to be inconsistent.

\subsubsection{Greedy Stability and Consistent Preferences}

It is easy to see that if every employer's preferences are consistent, then we can treat the problem as a stable marriage.

\begin{proposition}
Consider an affiliate matching marketplace where each employer is affiliate-agnostic, i.e. has expressed consistent preferences. The imposition of greedy stability onto this marketplace reduces the problem to stable marriage, and all results therefrom can be concluded.
\end{proposition}

\begin{proof}
By the definition of consistent preferences for the employer, the first element of every tuple in $\succ_e$ follows an order (that of $\overline{\succ}$). This means that under the definition of greedy blocking pair, we will have $(a,e)$ as a blocking pair of match $\mu$ only if $a \overline{\succ} \mu(e)$. Therefore, we can reduce the problem to that of stable marriage where the employers preferences are expressed as $\overline{\succ}$ and the applicant's preferences are expressed as $\prec_a$. 
\end{proof}

However, without consistency of the preferences, there may not be a stable match, i.e., the core may be empty. 

\subsubsection{Empty Core}

The affiliate matching problem with greedy stability may have an empty core, as seen in Example \ref{ex: empty}.

\begin{example}\label{ex: empty}
Consider the matching marketplace depicted in Figure \ref{fig:example marketplace} of three employers and three employer where each employer has only one affiliate (signified in the grey connections).
We observe that each employer does not admit consistent preferences. 
In this setting, every match has a greedy blocking pair.
The matching $\mu_1$ is blocked by greedy blocking pair $(a_1,e_2)$ which prefers $\mu_4$.
The matching $\mu_2$ is blocked by greedy blocking pair $(a_1,e_2)$ which prefers $\mu_4$.
The matching $\mu_3$ is blocked by greedy blocking pair $(a_1,e_3)$ which prefers $\mu_6$.
The matching $\mu_4$ is blocked by greedy blocking pair $(a_1,e_3)$ which prefers $\mu_6$.
The matching $\mu_5$ is blocked by greedy blocking pair $(a_3,e_3)$ which prefers $\mu_1$.
The matching $\mu_6$ is blocked by greedy blocking pair $(a_2,e_1)$ which prefers $\mu_5$.
\end{example}

\subsection{Integer-Program-based Approaches to Clearing Affiliate Matching Markets}\label{sec:stability-IP}

General clearing approaches that can take one of many objective functions, or one of many sets of constraints (such as one or more notions of stability) are necessary when exploring the fielding of a new matching market.  Thus, it is natural to look at integer programs, particularly with an eye towards real world implementation of this model where we may wish to maximize socially desirable outcomes, like number of agents with their top choice, some notion of fairness or equity, or other linearly expressed optimization objectives.

We proceed by giving an integer program of greedy stable affiliate matching. Let $Z$ be a binary matrix with $z_{ij}$ as an indicator variable of whether $a_i$ is matched to $e_j$. 

To enforce the definition of a match, we require that every agent is matched with at most one other agent.
\begin{align}
    \sum_{i} z_{ij} &\leq 1 \qquad\qquad \forall j\\
    \sum_{j} z_{ij} &\leq 1 \qquad\qquad \forall i
\end{align}

To ensure the match is stable, we need to require that for all pairs $(i,j)$, the tuple $(a_i,e_j)$ is not a blocking pair. 
So either, the pair is matched, i.e. $\mu(a_i)=e_j$ which happens if and only if $z_{ij}=1$, or they are not matched.
If they are not matched, then there are two options: (1) $a_i$ is matched to someone they prefer over $e_j$, or (2) there exists some applicant who $e_j$ ``prefers more"\footnote{This is in quotes because the preference language is not exclusively over agents.} given the current matchings of $e_j$'s affiliates. 

The second condition warrants some extra care. 
The condition for the employers in the definition of a blocking pair states: there exists $\gamma\in E^{r_e}$ such that $(a,\gamma)\succ_e(\mu(e),\mu(R_e))$.
Its logical negation is: $(a,\gamma)\prec_e(\mu(e),\mu(R_e))$ for all $\gamma\in E^{r_e}$.
As we are forming the linear program constraints, for a given $Z$, we can map the affiliates of $e$ and denote $Z(R_e)$. 
We note that $Z(R_e)\in E^{r_e}$ and is fixed for a given $Z$. 
Therefore, to test if $e$ breaks the blocking pair condition, we look for those $a'$ such that $(a,\gamma) \prec_e (a',Z(R_e))$ for all $\gamma$.
We now state the stability condition of the integer program: for all $(i,j)$,
\begin{align}
    z_{ij} + \sum_{\ell: e_j \prec_{a_i} e_\ell} z_{\ell j} + \sum_{\ell : \forall\gamma ,(a_i,\gamma)\prec_{e_j}(a_\ell,Z(R_{e_j}))} z_{i\ell} \geq 1.
\end{align}

\subsection{Strict Stability}

Intuitively, an employer may only wish to deviate from a match if they can find other agents among whom they all match and strictly improve their outcome with respect to their full preference profile $\succ$. Recall that the problem with $a_2$ and $e_1$ leaving the match of $\mu_1$ is that $e_1$ could either end up in a better position or a worse position, depending on what happens to its affiliates. 

In the following discussion, we introduce strict stability, which says that an employer wouldn't deviate from a match unless it is sure it would improve its outcome. In the running example of $(a_2, e_1)$, this means that $e_1$ wouldn't deviate from their match unless the final match of the entire system were $\mu_5$.

\ignore{\subsection{Strict Blocking Pair}

\begin{definition}
A match $\mu$ has a {\bf strict blocking pair} $(a,e)$ with respect to an alternate matching $\mu'$ such that the following conditions hold: (1) $\mu'(a)=e$, (2) $e \succ_a \mu(a)$, and (3) $(a,\mu'(R_e)) \succ_e (\mu(e),\mu(R_e))$.
\end{definition}

\begin{definition}
A matching $\mu$ is {\bf strictly stable} if there does not exists a strict blocking pair $(a,e)$ with respect to $\mu'$ such that $\mu'$ is also strictly stable. 
\end{definition}}

\begin{definition}\label{def: strict coalition}
A match $\mu$ has a {\bf strict blocking coalition} $(C_A, C_E)$ if there exists a match $\mu'$ where the following hold:
\begin{enumerate}
    \item $C_A\subset A$, $C_E\subset E$;
    \item for all $e\in C_E$, it is true that $R_e \subset C_A$, $\mu'(e)\in C_A$, and  $(\mu'(e),\mu'(R_e)) \succ_e (\mu(e),\mu(R_e))$; and
    \item for all $a\in C_A$, it is true that $\mu'(a) \succ_a \mu(a)$ and $\mu'(a)\in C_E$.
\end{enumerate}
\end{definition}

A direct result of this definition is that this coalition entirely matches with itself, and more importantly every affiliate of every employer is matched to another one of the employers in the coalition.  

Returning to our example of in Figure \ref{fig:example marketplace} where the system is currently matched under $\mu_1$. Recall that $e_1$ would like to consider coalition with $a_2$ in order to improve its outcome. However, in order to form a strict blocking coalition, the coalition would at least have to comprise of $a_1, a_2$, and $e_1, e_2$. One can check that $\mu_5$ is preferred to $\mu_1$ by all elements of this coalition, thus they have formed a strict blocking coalition. Note, that by doing so, $a_3$'s new outcome would move from their first choice to their last, and $e_3$ would also wind up with their last choice. 

We now state the obvious definition of strictly stable and show that Figure \ref{fig:example marketplace} is stable.

\begin{definition}
A match $\mu$ is {\bf strictly stable} if any coalition $(C_A,C_E)$ is not a strictly blocking coalition. 
\end{definition}

\begin{example}
The marketplace in Figure \ref{fig:example marketplace} has a strictly stable match: $\mu_1$. This match is unique because it is comprised of three coalitions already $(a_i, e_i)$ for $i=1,2,3$. The only other coalitions to consider are as follows.
The coalition $(\{a_1,a_2\},\{e_1,e_2\})$ is not blocking because $e_1$ does not prefer $\mu_3$ to $\mu_1$. 
The coalition $(\{a_1,a_3\},\{e_1,e_3\})$ is not blocking because $a_3$ does not prefer $\mu_6$ to $\mu_1$.
The coalition $(\{a_2,a_3\},\{e_2,e_3\})$ is not blocking because $e_3$ does not prefer $\mu_2$ to $\mu_1$.
\end{example}

\section{%
Designing Affiliate Matching Mechanisms: Desiderata and Challenges%
}\label{sec:challenges}

In this final section, we outline multi-disciplinary research challenges that should be addressed first as this topic develops. We believe that without a community discussion and consensus on the following topics, then there may not be a clear answer as to the form of mechanism to deploy for the affiliate matching problem---if such a mechanism is appropriate to deploy at all.
These challenges are loosely coupled, but exist with distinct motivations and implications. 
We conclude with a discussion of some additional ethical considerations that place affiliate matching in the context of economic mechanisms.

\begin{challenge}\label{rc: prestige}
Does an affiliate matching mechanism reify notions of prestige in a way that produces harm?
\end{challenge}

The importance and impact of prestige has been well-documented in the literature on labor markets~\citep[e.g.,][]{caplow1958academic,crane1965scientists,long1979entrance,burris2004academic,way2016gender, sauder2012status, clauset2015systematic}.
Those individuals who are able to obtain positions at prestigious employers have access to  distinctive resources which increase their ability to yield more desirable outcomes in their careers. These impacts often perpetuate themselves and become ingrained into unchangeable systems which operate to advance the interests of those at the top of a stratified system. 

\ignore{These hierarchical systems also appear to operate under social network effects. \citet{burris2004academic} argues that the social network created by insular PhD exchanges is one of the most important forms of social capital in the marketplace. They argue that  exchanging PhDs among institutions serves to produce and reinforce status separations within the hierarchy. Through this lens, we believe that Definition \ref{def: strict coalition} of a strict blocking coalition aligns well with this network effect. As a mechanism designer, it would be imperative that there not be a self-dealing coalition of employers and affiliates who leave a marketplace to advance themselves at the exclusion of others.}

Under \citet{bourdieu1986forms}'s conception of social and economic capital, the model which we created above could be used to enrich those who currently have some social capital with more forms of economic and social capital. Depending on your view, this is an inevitable outcome of humans and network effects, or this is a phenomenon which we can try to work against in the design of some matching mechanisms. Without a clear and thoughtful consensus on the appropriate goal of a central mechanism in this problem, any further work in this space could be pointless.

\begin{challenge}\label{rc: current}
How do current affiliate marketplaces operate?
\end{challenge}

In order to understand truly if an affiliate matching mechanism is appropriate, there must be a thorough study of current marketplaces. We must first understand how the affiliate marketplaces function and what their inefficiencies are. Scholarship in this research challenge will bring about clarity with respect to how to properly incorporate any proposed, appropriate matching mechanism. 

Work, like that cited in Research Challenge \ref{rc: prestige} include some answers to these questions. Furthering scholarship~\citep[see, e.g.,][]{wang2003incremental,boutilier2006constraint,domshlak2011preferences,vayanos2020active} can investigate how preference elicitation can be efficiently performed for the complex decision space of employers across $A\times E^{r_e}$.  We also found evidence of strategic behavior on the part of employers in our survey (e.g., from Respondent \#6, ``I will settle for my second choice so he can go to the top tiered school. Worth it I think.'')  Indeed, progress here should motivate extensions to models traditional concerns for the computational social choice community---specifically, preference elicitation and aggregation, and incentive properties~\citep[e.g.,][]{dietrich2007strategy,baumeister2013computational,brandt2016handbook}.  

\begin{challenge}\label{rc: care}
How much do employers care about their affiliates? How much \emph{should} they care?
\end{challenge}

Respondents in our survey had differential perspectives on if it was appropriate to care about the match of their affiliates. Respondent \#133 indicated, ``If it cannot be fair, it must be the least unfair.  The most unfair would be to downgrade the the top two and advantage the lowest.'' While Respondent \#29 said ``Even if Ryan isn't the top candidate, he should be given the chance at the top school.'' Understanding and balancing these divergent views will be crucial to implementation and adoption of any designed central clearinghouse. In-depth study should involve domain experts in labor markets. 

\begin{challenge}\label{rc: stability}
What is the right definition of stability?
\end{challenge}

As we saw in Section~\ref{sec:stability}, differences in the definition of stability can lead to different ``acceptable'' outcomes.
As is true in all designing of matching mechanisms, stability is used as way to enforce participation in a marketplace. The logic is that deviation from a mechanism's stable outcome would be weakly worse for a participant.
These marketplaces operate under the implicit social contract between agents that adherence to the mechanism's output will be mutually beneficial in terms of overall harm reduction.


~\\
It is our hope that by studying and understanding these research challenges more, a community of scholars can form around this  affiliate matching problem and can engage with stakeholders in potential labor markets to determine the appropriateness of a matching mechanism.

\subsection{Ethical Impact of our Work}

In addition to the ethical questions that are brought about by the above research challenges, we see other important discussion points. Making socially- and morally-laden decisions are necessary steps when designing economic mechanisms; indeed, the allocation of scarce resources, or matching agents to other agents or opportunities, often comes with unavoidable creation of winners and losers.
One way to enforce socially desirable outcomes in a matching mechanism is by careful definition and analysis of desiderata, such as stability, of a mechanism. Stability in particular yields one form of fairness---it eliminates justifiable envy---which we as mechanism designers can define in a way to produce outcomes that align with social and moral values of a society. 

Additionally, the way that prestige operates in affiliate marketplaces could potentially be viewed as deleterious---enriching those who are already privileged at the expense of those with fewer means and resources. At issue will be whether a course of study in affiliate matching is about understanding the current practices for the sake of understanding, or if the ultimate goal is reform. If it is the latter, any reform has unintended consequences which could ultimately cause more harm than any intended benefit. We hope that this work serves to contribute in a positive direction towards understanding a current system and its inequities, but we also acknowledge that it might ultimately serve to solidify some of these very ideas. 

\section*{Acknowledgements}
This research was supported in part by NSF CAREER Award IIS-1846237, NIST MSE Award \#20126334, DARPA GARD \#HR00112020007, DARPA SI3-CMD \#S4761, DoD WHS Award \#HQ003420F0035, and a Google Faculty Research Award.

We would like to thank Marina Knittel and Duncan McElfresh for their input and thoughtful review of the manuscript. We also thank the anonymous participants of our survey for their time and contributions to this work. 

\bibliographystyle{plainnat}  
\bibliography{refs}

\appendix
\section{Survey Protocol} \label{appendix: protocol}

{\textbf{Faculty Hiring Program}}

~

The design of this survey is aimed at understanding how you make
decisions with different competing priorities. You will be exploring
this concept in the setting of a hiring market for new faculty
professors. The Survey will have two parts: (1) familiarization with
faculty hiring, and (2) answering questions about your preferences. We
begin with the familiarization part now.

Consider a hiring market such as this one with three applicants (Ryan,
Alex, and Taylor), and three universities (Bear Mountain, Littlewood,
and West Shores).

\includegraphics[width=.4\linewidth]{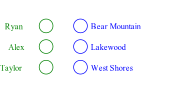}
\\
Imagine that \textbf{you are Bear Mountain University}, and you
performed an evaluation of the applicants. You decided that you liked
Alex more than you liked Taylor, and you liked Taylor more than you
liked Ryan.~

\adjustbox{max width=\linewidth}{%
\begin{tabular}{c c c}
\toprule
{{Top Tier Candidates}} & {{Middle Tier Candidates}} & {{Bottom Tier
Candidates}}\\
Alex & Taylor & Ryan\\
\bottomrule
\end{tabular}
}

\hfill\break

You could depict that preference as:

~

\textbf{{First Choice}}

\includegraphics[width=.3\linewidth]{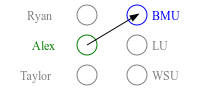}

\textbf{{Second Choice}}

\includegraphics[width=.3\linewidth]{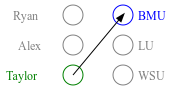}

\textbf{{Third Choice}}

\includegraphics[width=.3\linewidth]{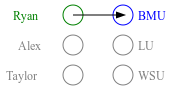}

\noindent\makebox[\linewidth]{\rule{\linewidth}{0.4pt}}

To test your comprehension of the previous setting, can you now do the
matchings yourself? These are intuitive and should be easy to complete.~

~

Assume \textbf{you are Bear Mountain University}. Assume this time that
after you review the applicants, you place them in these tiers:

\adjustbox{max width=\linewidth}{%
\begin{tabular}{c c c}
\toprule
{{Top Tier Candidates}} & {{Middle Tier Candidates}} & {{Bottom Tier
Candidates}}\\
Ryan & Alex & Taylor\\
\bottomrule
\end{tabular}
}

Then what is your ranking of the following options?~

\textbf{Assuming you believe the above}, rank these outcomes from your
most preferred (1) to your least preferred (3).
\includegraphics[width=.3\linewidth]{imgs/alex.png}
\includegraphics[width=.3\linewidth]{imgs/ryan.png}
\includegraphics[width=.3\linewidth]{imgs/taylor.png}

~

Assume \textbf{you are Bear Mountain University}. Assume this time that
after you review the applicants, you place them in these tiers:

\adjustbox{max width=\linewidth}{%
\begin{tabular}{c c c}
\toprule
{{Top Tier Candidates}} & {{Middle Tier Candidates}} & {{Bottom Tier
Candidates}}\\
Alex & Taylor & Ryan\\
\bottomrule
\end{tabular}
}

Then what is your ranking of the following options?~

\textbf{Assuming you believe the above}, rank these outcomes from your
most preferred (1) to your least preferred (3).
\includegraphics[width=.3\linewidth]{imgs/alex.png}
\includegraphics[width=.3\linewidth]{imgs/ryan.png}
\includegraphics[width=.3\linewidth]{imgs/taylor.png}

~

Assume \textbf{you are Bear Mountain University}. Assume this time that
after you review the applicants, you place them in these tiers:

\adjustbox{max width=\linewidth}{%
\begin{tabular}{c c c}
\toprule
{{Top Tier Candidates}} & {{Middle Tier Candidates}} & {{Bottom Tier
Candidates}}\\
Taylor & Alex & Ran\\
\bottomrule
\end{tabular}
}

Then what is your ranking of the following options?~

\textbf{Assuming you believe the above}, rank these outcomes from your
most preferred (1) to your least preferred (3).
\includegraphics[width=.3\linewidth]{imgs/alex.png}
\includegraphics[width=.3\linewidth]{imgs/ryan.png}
\includegraphics[width=.3\linewidth]{imgs/taylor.png}

\noindent\makebox[\linewidth]{\rule{\linewidth}{0.4pt}}

Awesome! Now onto the second part of the Survey. Since you understand
the basic faculty hiring setting, let us introduce another layer of
complexity.~

In faculty hiring, the applicants are affiliated with a university based
off where they earned their PhD. What this means is that universities
also care about where their student gets a job.

\textbf{Assume that you are Bear Mountain University and your student is
Ryan}. You then have the following five options of matchings:

You are matched with Ryan.

\includegraphics[width=.3\linewidth]{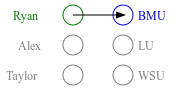}

You are matched with Alex; Ryan is matched with Littlewood University.

\includegraphics[width=.3\linewidth]{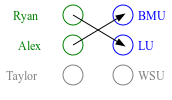}

You are matched with Taylor; Ryan is matched with Littlewood University.

\includegraphics[width=.3\linewidth]{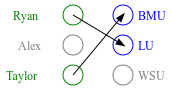}

You are matched with Alex; Ryan is matched with West Shores University.

\includegraphics[width=.3\linewidth]{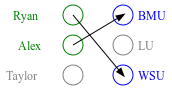}

You are matched with Taylor; Ryan is matched with West Shores University.

\includegraphics[width=.3\linewidth]{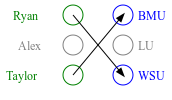}

In the remainder of the survey, we will ask you to express your preferences over these five options under different settings of which applicants you like best and what you think of the different universities.

For the remainder of the survey, you will get to decide how you would like to balance your own interest in being matched with the best candidates possible, and where Ryan should be matched. There is no right answer -- it is up to you about how you balance these two things.

\noindent\makebox[\linewidth]{\rule{\linewidth}{0.4pt}}

{\it Randomly assigned either of the two following prompts:}
\\

For the remainder of the survey, you will get to decide how you would like to balance your own interest in being matched with the best candidates possible, and where Ryan should be matched. There is no right answer -- it is up to you about how you balance these two things.

{\it or}

In faculty hiring, it is common to want your student to be placed at the best university possible. This is often because you as a university will be perceived as a better university if your students get jobs at top tier schools.

So, keep in mind that if Ryan is placed at a higher ranked university, then your university will be viewed better and could hire better candidates in the future. While you want Ryan to be matched with the best school, you must balance this priority with your competing priority that you want the best candidate possible. There is no right answer -- it is up to you about how you balance these two things.
\\

We will ask you to express your preferences to 8 scenarios. When you are ready, please continue to Part 2.

\noindent\makebox[\linewidth]{\rule{\linewidth}{0.4pt}}

{\it Display this question 8 times with the ordering in the tables randomized without replacement.}
\\

Assume that you are Bear Mountain University and Ryan is your student.~

Assume that you have evaluated the candidates and you decided to place
the applicants in the following tiers. You think that Alex is the best
candidate for your job, then Taylor is second best, and Ryan is least
qualified.~

\adjustbox{max width=\linewidth}{%
\begin{tabular}{c c c}
\toprule
{{Top Tier Candidates}} & {{Middle Tier Candidates}} & {{Bottom Tier
Candidates}}\\
Alex & Taylor & Ryan (your student)\\
\bottomrule
\end{tabular}
}

\hfill\break
Assume that you have the following beliefs about the schools, based off
of the school's ranking. You think that Littlewood (LU) is the best
school, then you think your school (BMU) is in the middle tier of
schools, and West Shores (WSU) is in the bottom tier.

\adjustbox{max width=\linewidth}{%
\begin{tabular}{c c c}
\toprule
{{Top Tier Schools}} & {{Middle Tier Schools}} & {{Bottom Tier
Schools}}\\
LU & BMU (your school) & WSU\\
\bottomrule
\end{tabular}
}
\\

{\bf Assuming you believe the above}, rank these outcomes from your most preferred (1) to your least preferred (5).

\noindent\includegraphics[width=.3\linewidth]{imgs/qual_1.png}
\includegraphics[width=.3\linewidth]{imgs/qual_2.png}
\includegraphics[width=.3\linewidth]{imgs/qual_3.png}\\
\includegraphics[width=.3\linewidth]{imgs/qual_4.png}
\includegraphics[width=.3\linewidth]{imgs/qual_5.png}

Can you describe the procedure you used to rank these? Why did you use this procedure?

\end{document}